\RequirePackage{amsthm}
\documentclass[sn-mathphys-num]{sn-jnl}

\usepackage{graphicx}%
\usepackage{multirow}%
\usepackage{amsmath,amssymb,amsfonts}%
\usepackage{amsthm}%
\usepackage{mathrsfs}%
\usepackage[title]{appendix}%
\usepackage{xcolor}%
\usepackage{textcomp}%
\usepackage{manyfoot}%
\usepackage{booktabs}%
\usepackage{algorithm}%
\usepackage{algorithmicx}%
\usepackage{algpseudocode}%
\usepackage{listings}%

\usepackage{cleveref}
\usepackage{dsfont, stmaryrd}
\usepackage{nicematrix}
\usepackage{enumerate}
\usepackage{mlmodern}

\theoremstyle{thmstyleone}%
\newtheorem{theorem}{Theorem}
\newtheorem{proposition}[theorem]{Proposition}%
\newtheorem{corollary}{Corollary}

\theoremstyle{thmstyletwo}%
\newtheorem{example}{Example}%
\newtheorem{remark}{Remark}%

\theoremstyle{thmstylethree}%
\newtheorem{definition}{Definition}%

\usepackage{enumerate}
\usepackage{nicematrix}
\usepackage{dsfont, stmaryrd}
\newcommand{\FF}{\mathds{F}}
\newcommand{\ZZ}{\mathds{Z}}
\newcommand{\KK}{ \mathds{K} }
\newcommand{\kk}{\mathds{k}}
\newcommand{\Mat}{\mathds{M}}
\newcommand{\calA}{\mathcal{A}}
\newcommand{\calI}{\mathcal{I}}
\newcommand{\ord}{\operatorname{ord}}
\newcommand{\diag}{\operatorname{diag}}
\newcommand{\Aut}{\operatorname{Aut}}
\newcommand{\intrange}[2]{\llbracket #1,#2 \rrbracket}
\newcommand{\End}{\operatorname{End}}

\newcommand{\supp}{\operatorname{supp}}
\newcommand{\Supp}{\operatorname{Supp}}
\newcommand{\wt}[1]{\operatorname{wt}\left( #1 \right)}
\newcommand{\dist}{\operatorname{dist}}

\newcommand{\boldB}{\mathbf{B}}
\newcommand{\Ext}{\operatorname{Ext}}

\newcommand{\Mod}[1]{\ (\operatorname{mod}\ #1)}

\usepackage{orcidlink}
\begin{document}
\title[On Codes from Split Metacyclic Groups]{On Codes from Split Metacyclic Groups}


\author[]{\fnm{Kirill} \sur{Vedenev}\;\orcidlink{0000-0002-7893-655X}\footnote{E-mail: \texttt{vedenevk@gmail.com}}}



\abstract{The paper presents a comprehensive study of group codes from non-abelian split metacyclic group algebras. We derive an explicit Wedderburn-like decomposition of finite split metacyclic group algebras over fields with characteristic coprime to the group order. Utilizing this decomposition, we develop a systematic theory of metacyclic codes, providing their algebraic description and proving that they can be viewed as generalized concatenated codes with cyclic inner codes and skew quasi-cyclic outer codes. We establish bounds on the minimum distance of metacyclic codes and investigate the class of induced codes. Furthermore, we show the feasibility of constructing a partial key-recovery attack against certain McEliece-type cryptosystems based on metacyclic codes by exploiting their generalized concatenated structure.}




\maketitle

\section{Introduction}
Let $G$ be a finite group and $R$ be a ring with unity. The set $RG$ of all formal linear combinations of the form
$\sum_{g \in G} \lambda_g g$, where $\lambda_g \in R$,
equipped with the following operations of addition, multiplication, and multiplication by a scalar:
    \[ 
        \left( \sum_{g \in G} \lambda_g g \right) +
        \left( \sum_{g \in G} \mu_g g \right) =
        \sum_{g \in g} (\lambda_g + \mu_g) g,
    \]
    \[ 
        \left( \sum_{g \in G} \lambda_g g \right) 
        \left( \sum_{g \in G} \mu_g g \right) =
        \sum_{g \in g} \Bigg( \sum_{\substack{h, h' \in G \\ hh' = g}} \lambda_h \mu_{h'} \Bigg) g = \sum_{g \in G} \left( \sum_{h \in G} \lambda_h \mu_{h^{-1}g}  \right) g,
    \]
    \[  \mu \left( \sum_{g \in G} \lambda_g g \right) =  \sum_{g \in G} ( \mu \lambda_g ) g, \quad  \left( \sum_{g \in G} \lambda_g g \right) \mu = \sum_{g \in G} (\lambda_g \mu) g, \]
is called the \emph{group ring} of $G$ over $R$. If $R$ is commutative, then $RG$ is also called the \emph{group algebra} of $G$ over $R$ \cite[\S 3.2]{Miles2002}. 
There is a natural embedding of $G$ into $RG$ given by $x \mapsto \sum_{g \in G} \lambda_g g$, where $\lambda_x = 1$ and $\lambda_g = 0$ for all $g \neq x$. Similarly, there is a natural embedding of $R$ into $RG$ via $r \mapsto re$. So, both $G$ and $R$ can be regarded as subsets of $RG$.

Given an element $u = \sum_{g \in G} u_g g \in RG$, the \emph{support} of $u$ is defined by
\[ \supp(u) = \left\{ g \in G \mid u_g \neq 0 \right\}, \]
and the \emph{Hamming weight} of $u$ is defined by $\wt{u} = |\supp(u)|$. \emph{The Hamming distance} on $RG$ is given by $\dist_H(u,v) = \wt{u - v}$ for $u, v \in RG$.

For a finite field $\FF_q$ of cardinality $q$ and a finite group $G$, any left (respectively, right) ideal $C$ of the group ring $\FF_qG$ endowed with the Hamming distance $\dist_H$ is called a \emph{left (resp. right) group code} or a \emph{left (resp. right) $G$-code} \cite{Willems2021}. 
Note that the anti-automorphism of $\FF_qG$ given by
\begin{equation} \label{d:eq:anti}
    u = \sum_{g \in G} u_g g \mapsto u^* = \sum_{g \in G} u_{g^{-1}} g 
\end{equation} 
(see \cite[Proposition 3.2.11]{Miles2002}) establishes a one-to-one correspondence between left and right $G$-codes. In other words, if $C$ is a left (resp. right) $G$-code, then $C^* = \{ c^* \mid c \in C \}$ is a right (resp. left) $G$-code. Therefore, we will mainly focus on left $G$-codes, which will be simply referred to as $G$-codes unless otherwise specified.

The group codes, introduced independently by S.~Berman in \cite{Berman1967} and F. MacWilliams in \cite{MacWilliams1970}, form a powerful class of linear codes that possess many desirable properties, including efficient encoding and decoding, as well as the applicability of algebraic methods for studying them. Abelian codes, i.e., codes from abelian group algebras, have been studied in some depth, while there are not many systematic results known about non-abelian codes. Non-abelian codes are particularly interesting due to their possible applications in code-based cryptography, since their complex algebraic structure was conjectured to improve the security of code-based cryptosystems and reduce public key sizes \cite{Deundyak2015,Deundyak2016}.

Being the simplest class of non-abelian groups, split metacyclic groups $G_{n,m,r}$, which are defined by the following presentation
\[
    G_{n,m,r} = \left< a, b \mid a^n = b^m = e, ba = a^r b \right>, 
\]
where $r^m \equiv 1 \Mod{n}$, are natural choice for studying non-abelian codes. As demonstrated in \cite{Vedenev2018, Vedenev2019, Vedenev2020, Vedenev2020-2} in the case of dihedral groups, the Wedderburn decomposition of a group algebra into a direct sum of matrix algebras turns out to be a very powerful and convenient tool for studying group codes. However, the problem of explicitly constructing such decompositions is very non-trivial.

\emph{Contribution.} The contribution of this paper is twofold. First, an explicit Wedderburn-like decomposition of finite split metacyclic group algebras is obtained. Second, a systematic theory of split metacyclic codes is developed by leveraging this decomposition. Specifically, it is proved that metacyclic codes can be viewed as generalized concatenated codes, with inner codes being cyclic codes and outer codes being skew quasi-cyclic codes. In addition, the class of induced codes is studied, and estimates of the main parameters of metacyclic codes are obtained. Finally, the possibility of building a partial key-recovery attack against certain metacyclic code-based McEliece-type cryptosystems is demonstrated.

\emph{Organization.} Section \ref{sec:preliminaries} provides the necessary preliminaries. In Section \ref{sec:structure}, a decomposition for finite split metacyclic group algebras in the case where $\gcd(q,n)=1$ is obtained. In Section \ref{sec:codes}, the algebraic description of metacyclic codes is given, and their concatenated structure is studied. Additionally, these results are used to derive a lower bound on the minimum distance of metacyclic codes and to build partial key-recovery attacks against cryptosystems based on some metacyclic codes. Section \ref{sec:induced} provides an algebraic description of induced codes and derives a lower bound on the minimum distance of metacyclic codes by leveraging induced codes. Finally, Section \ref{sec:concl} concludes the paper.

\emph{Prior and Related Works.}
In 1994, R. Sabin \cite{Sabin1994} showed that some quasi-cyclic codes can be viewed as ideals of metacyclic group algebras and discovered that several such codes have minimum distances equal to those of the best-known linear codes. In 1995, R. Sabin and S. Lomonaco \cite{Sabin1995} proved that central codes, i.e., two-sided ideals, from semisimple metacyclic group algebras are combinatorially equivalent to abelian codes, and provided several examples of good non-central codes obtained by leveraging group representations. Additionally, they described an algorithm for finding irreducible representations in the case when the ambient field $\FF_q$ contains all $n$-th roots of unity.

In 2016, S. Assuena and C.P. Miles \cite{Assuena2016} considered semisimple non-abelian metacyclic group algebras and described their primitive central idempotents in the case when the order of $G$ equals $p^m l^n$, where $p$ and $l$ are different prime numbers. In their recent works \cite{assuena2019, Assuena2020}, S. Assuena and C.P. Miles proposed constructions of some non-central codes from metacyclic group algebras by leveraging idempotents derived from subgroups, with parameters of some of those codes matching the best known linear codes.

In 2007, O. Broche and A. del Rio \cite{Broche2007} proposed a computational method for describing the Wedderburn decomposition and the primitive central idempotents of a semisimple finite group algebra of an abelian-by-supersolvable group $G$ from certain pairs of subgroups of $G$. Building upon this work, in 2014, G. Olteanu and V. Gelder \cite{Olteanu2014} proposed algorithms to construct minimal left group codes and showed that their main result can be applied to metacyclic groups of the form $G_{q^m, p^n, r} = C_{q^{m}} \leftthreetimes C_{p^n}$ with $C_{p^n}$ acting faithfully on $C_{q^m}$, where $p$ and $q$ are different primes and the field size $s$ is coprime to $p$ and $q$.

In 2015, F.~E.~Brochero Martinez \cite{Martinez2015} obtained an explicit Wedderburn decomposition of the semisimple dihedral group algebras $\FF_qD_{2n}$, where $D_{2n} = G_{n, 2, -1}$. In \cite{Vedenev2019,Vedenev2020, Vedenev2020-2,Vedenev2021}, the systematic theory of dihedral codes was developed in terms of this decomposition.

In 2020, Gao et al. \cite{Gao2020} generalized the results of \cite{Martinez2015} by obtaining an explicit Wedderburn decomposition for $\FF_qG_{n,2,r}$, where $G_{n,2,r}$ is defined as above and $r^2 \equiv 1 \pmod{n}$. In addition, Gao et al. \cite{Gao2020} described some linear complementary dual (LCD) codes from these group algebras.

In 2016, Cao et al. \cite{Cao2016} studied the concatenated structure of dihedral codes leveraging only finite field theory and basic theory of cyclic codes and skew cyclic codes. Using similar methods, Cao et al. \cite{Cao2016-2} proved the concatenated structure of codes from a class of metacyclic groups of the form $G_{n,3,r}$. In 2022, Cao et al. \cite{Cao2022} refined the results of \cite{Cao2016} and determined all distinct Euclidean LCD codes and Euclidean self-orthogonal dihedral codes in terms of their concatenated structure.

In 2021, M.~Borello and A.~Jamous \cite{Borello2021} derived a BCH-like lower bound on the minimum distance of dihedral codes by viewing dihedral codes as subcodes of expanded cyclic codes over field extensions. Note that a similar technique was leveraged by K.~Lally in \cite{Lally2003} for deriving the minimum distance bound for quasi-cyclic codes.

Note that the prior works considered metacyclic codes with serious restrictions on the parameters $n, m, r, q$ and/or focused on developing certain good examples of such codes.

\section{Preliminaries}
\label{sec:preliminaries}
\paragraph{Notation} 
We denote the finite field of size $q$ as $\FF_q$. The ring of polynomials over $\FF_q$ is denoted by $\FF_q[x]$, with $\FF_q[x]_n$ denoting the set of polynomials of degree $n$ and $\FF_q[x]_{< n}$ denoting the set of polynomials of degree less than $n$. Given an irreducible polynomial $g(x)$ over $\FF_q[x]$ with a root $\gamma$ in the splitting field of $\FF_q$, we denote the extension of $\FF_q$ with $\gamma$ by $\FF_q[\gamma] \simeq \FF_q[x]/(g(x))$. The notation $\intrange{a}{b}$, where $a, b \in \ZZ$, stands for the set $\{ i \in \ZZ \mid a \leq i \leq b \}$. We denote the identity map on a set $S$ by $\mathrm{id}_S$, and the identity $(m \times m)$-matrix is denoted by $E_m$.
\paragraph{Skew group algebras} 
Skew group algebras are a further generalization of usual group algebras, with their construction being somewhat analogous to the semidirect product of groups. Let $G$ be a finite group, $\KK$ be a field, and $\theta: G \to \Aut(\KK)$ be a group homomorphism. The \emph{skew group algebra} $\KK \ast_{\theta} G$ of $G$ over $\KK$ is the set of formal sums 
\[ 
    \KK \ast_{\theta} G = \left\{ \sum_{g \in G} a_g g \mid a_g \in \KK \right\},
\]
with addition defined componentwise, and multiplication distributively extending the following rule:
\[ 
    g \lambda = \left( \theta(g) (\lambda) \right)  g \quad \text{for } g \in G \text{ and } \lambda \in \KK. 
\]
Consequently, the multiplication of two elements in $\KK \ast_{\theta} G$ is given by
\begin{align*}
    \left( \sum_{g \in G} a_g g \right) \cdot \left( \sum_{g \in G} b_g g \right) &= 
    \sum_{g \in G} \Bigg( \sum_{\substack{h_1, h_2 \in G \\ h_1 h_2 = g}} a_{h_1} \left( \theta(h_1)(b_{h_2}) \right) \Bigg) g = \\ &= 
    \sum_{g \in G} \left( \sum_{h \in G} a_{h} \left( \theta(h)(b_{h^{-1} g}) \right) \right) g.
\end{align*} 
The field $\KK$ and the group $G$ are naturally embedded into $\KK \ast_{\theta} G$ via the maps $\lambda \mapsto \lambda e$ and $g \mapsto 1g$, respectively. Thus, the multiplication by scalars is defined as an instance of generic multiplication.

In the following, matrix rings over skew group algebras will appear as direct summands in the Wedderburn-like decomposition of metacyclic group algebras. Thus it is essential to gain a deeper understanding of their algebraic structure.

The following proposition, essentially proven in \cite[Corollary 29.8]{Reiner2003}, establishes the isomorphism between skew group algebras and matrix algebras in certain cases. For the sake of convenience and completeness, we also provide an alternative proof. 

\begin{proposition} \label{prop:skew}
    Let $\KK$ be a field, $G$ be a finite group, $\theta: G \to \Aut(\KK)$ be a group monomorphism. Then $\KK \ast_{\theta} G$ is isomorphic to $\Mat_{|G|}(\kk)$, where 
    \[ 
        \kk = \left\{ \mu \in \KK \mid \forall g \in G \;\: \theta(g)(\mu) = \mu \right\}
    \] is the fixed field of $G$.
\end{proposition}
\begin{proof}
    Let $\sigma: \KK \ast_{\theta} G \to \End_\kk(\KK)$ be $\kk$-algebra homomorphism defined by
    \[ 
        \sigma \left( \sum_{g \in G} \lambda_g g \right) =
        \sum_{g \in G} \lambda_g \theta(g).
    \]
    First, we show that $\sigma$ is injective. Indeed, assume $\sum_{g \in G} \lambda_g g \in \ker(\sigma)$; then, since each $\theta(g) \in \Aut(\KK)$ can be considered as a multiplicative character on $\KK^{*}$, by the linear independence of characters theorem (see  \cite[\S VI.4]{Lang}), we obtain $\lambda_g = 0$ for all $g \in G$.

    We also have $[\KK:\kk] = |G|$ (see Theorem 1.8 of \cite[\S VI.1]{Lang}). Hence
    \[
        \dim_{\kk} \left( \End_\kk(\KK) \right) = |G|^2 = \dim_{\kk} \left( \KK \ast_{\theta} G \right).
    \]
    Thus, $\sigma$ is $\kk$-algebra isomorphism. Since $\End_\kk(\KK) \simeq \Mat_{|G|}(\kk)$, the proof is complete.
\end{proof}

\begin{corollary} \label{col:skew2}
    Let $G = H_1 \times H_2$. Let $\theta|_{ \{ e \} \times H_2 }$ be trivial, $\tilde{\theta} = \theta|_{ H_1 \times \{ e \} }$ be injective. Then
    \[ 
        \KK \ast_{\theta} G \simeq \left( \KK \ast_{\tilde{\theta}} H_1 \right) H_2 \simeq \left( \Mat_{|H_1|} (\kk) \right) H_2 \simeq \Mat_{|H_1|}(\kk H_2),  
    \]
    where
    \( 
        \kk = \left\{ \mu \in \KK \mid \forall h \in H_1 \;\: \tilde{\theta}(h)(\mu) = \mu \right\}. 
    \)    
\end{corollary}

\section{Structure of finite metacyclic group algebras}
\label{sec:structure}
In this section, we obtain an exlicit Wedderburn-like decomposition of the split metacyclic group algebras in the case $\gcd(q, n) = 1$.
\emph{Hereinafter in this paper, we assume that $\gcd(q,n) = 1$ and $r \not\equiv 1 \Mod{n}$. In addition, $A$ and $B$ stand for the cyclic subgroups of $G_{n,m,r}$ generated by $a$ and $b$, respectively.}

Before presenting the main result of this section, we develop some necessary preliminary results on factorization of $x^n - 1$. Recall that the $d$-th cyclotomic polynomial $Q_d(x)$ is defined as the polynomial whose roots are the primitive $d$-th roots of unity in some extension field of $\FF_q$. Additionally, as is well-known, $x^n-1 = \prod_{d \mid n} Q_d(x)$, and hence any irreducible factor of $x^n-1$ is a divisor of some $Q_d(x)$. Let $t_d$ denote the smallest positive integer such that $q^{t_d} \equiv 1 \Mod{d}$.

Given a monic irreducible factor $f(x)$ of $Q_d(x)$, its \emph{$r$-reciprocal polynomial $f^{(r)}(x)$} is defined as the monic minimal polynomial of $\beta^r$, where $\beta$ denotes a root of $f(x)$ in some extension of $\FF_q$. The polynomial $f(x)$ is said to be \emph{r-self-reciprocal} if and only if $f(x) = f^{(r)}(x)$.

Since $\gcd(r,n) = 1$, it follows that $\ord(\beta^r) = \ord(\beta) = d$, implying that $f^{(r)}(x) \mid Q_d(x)$. Additionally, one can easily note that $f(x)$ and $f^{(r)}(x)$ have factorizations of the form
\begin{equation} \label{eq:fr}
    f(x) = \prod_{j=0}^{t_d-1} (x - \beta^{q^j}), \quad\quad
    f^{(r)}(x) = \prod_{j=0}^{t_d-1} (x - \beta^{rq^j}),
\end{equation}
over the splitting field of $Q_d(x)$, respectively.

Let $\mathcal{D}_{\FF_q}(g)$ denote the set of all irreducible divisors of a polynomial $g$ over $\FF_q$. Define an action of $B$ on $\mathcal{D}(x^n-1)$ as follows:
\[ 
    b^j.f(x) = f^{(r^j)}(x) \quad \text{for $b^j \in B$ and $f(x) \in \mathcal{D}(x^n-1)$}. 
\]
This is indeed a group action since $\left( f^{(r^i)} \right)^{(r^j)}(x) = f^{(r^{i+j})}(x)$ and $f^{(r^m)}(x) = f(x)$. Now, let
\begin{itemize}
    \item $O_1, \dots, O_{\omega}$ be the set of orbits of $\mathcal{D}(x^n-1)$ under this action;
    \item $f_1, \dots, f_{\omega}$ be a system of representatives for $O_1, \dots, O_{\omega}$;
    \item $\alpha_1, \dots, \alpha_n$ be roots of $f_1, \dots, f_{\omega}$ in some extensions of $\FF_q$, respectively;
    \item $B_1, \dots, B_{\omega}$ be stabilizers of $f_1, \dots, f_{\omega}$, respectively;
    \item $s_i = |O_i|$ and $u_i = |B_i| = m/s_i$ for $i \in \intrange{1}{\omega}$.
\end{itemize}
One can easily note that $B_i = \left\{ b^j \in B \mid b^j.f_i(x) = f_i(x) \right\}$ is a cyclic subgoup of $B$ of order $u_i$, and hence $B_i = \left< b^{s_i} \right>$. To simplify the notation, we will denote its generator $b^{s_i}$ by $h_i$.

Let $d$ be such that $f_i(x) \mid Q_d(x)$ (or equivalently, $\ord(\alpha_i) = d$). Since $b^{s_i}.f_i(x) = f_i(x)$, it follows that $f_i(x)$ is $r^{s_i}$-self-reciprocal polynomial and, consequently, $\alpha_i^{r^{s_i}}$ is a root of $f_i(x)$. This is possible if and only if either $\alpha_i^{r^{s_i}} = \alpha_i$ or $\alpha_i^{r^{s_i}} = \alpha_i^{q^{k}}$ for some $k$ (see \eqref{eq:fr}). In other words, there exists a positive integer $k$ such that $q^k \equiv r^{s_i} \Mod{d}$. 

Given the above, by $\theta_i: B_i \to \Aut(\FF_q[\alpha_i])$ we denote a group homomorphism defined by generator $h_i$ of $B_i$ as
\[ 
    \theta_i(h_i): P(\alpha_i) \mapsto \left( P(\alpha_i) \right)^{q^k} = P\left( \alpha_i^{q^k} \right),    
\]
where $P(\alpha_i) \in \FF_q[\alpha_i]$. Note that $q^k \equiv r^{s_i} \Mod{d}$ and $\ord(\alpha_i) = d$ imply $\alpha_i^{q^k}  = \alpha_i^{r^{s_i}}$.

The following theorem provides an explicit decomposition of finite split metacyclic group algebras, given the factorization of $x^n-1$ and the group action defined above.
\begin{theorem} \label{theor:decomp}
    Let $gcd(n,q) = 1$. The group algebra $\FF_q G_{n,m,r}$ has a decomposition of the following form:
    \begin{equation} \label{eq:decomp}
        \FF_q G_{n,m,r} \overset{}{\simeq} \bigoplus_{i=1}^{\omega} \calA_i, \quad \text{where } \calA_i = \Mat_{s_i}(\FF_q[\alpha_i] \ast_{\theta_i} B_i).
    \end{equation} 
    Morover, the isomorphism is given by $\tau = \bigoplus_{i=1}^{\omega} \tau_i$, where the homomorphisms $\tau_i: \FF_qG_{n,m,r} \to \calA_i$ are defined on generators of $G_{n,m,r}$ as follows:
    \begin{enumerate}[(i)]
        \item if $s_i = 1$ (and hence $\calA_i = \FF_q[\alpha_i] \ast_{\theta_i} B$):
        \[ \tau(a) = \alpha_i, \quad \tau_i(b) = h_i = b, \]
        \item if $s_i \neq 1$:
        \[
            \tau(a) = \diag\left( \alpha_i, \alpha_i^r, \alpha_i^{r^2}, \dots, \alpha_i^{r^{s_i-1}} \right), \quad
            \tau(b) =  
            \begin{bNiceArray}{c|cw{c}{1cm}c}[margin]
                0 & \Block{3-3}<\Large>{E_{s_i - 1}} & & \\
                \Vdots & & & \\
                0 & & & \\ \hline
                h_i & 0 & \Cdots & 0
            \end{bNiceArray}.
        \]
    \end{enumerate}
\end{theorem}
\begin{proof}
    First, we verify that $\tau_i$, $i \in \intrange{1}{\omega}$ are indeed homomorphisms. To do this, it is sufficient to check that the defining relations of $G_{n,m,r}$ hold for the images of its generators. In the case (i), we have
    \[ 
        \left( \tau_i(a) \right)^n = \alpha_i^n = 1, \quad \left( \tau_i(b) \right)^m = b^m = 1,
    \] \[
        \tau_i(b) \tau_i(a) = b \alpha_i = \theta_i(\alpha_i) b = \alpha_i^{r} b = \left( \tau_i(a) \right)^r \tau_i(b).
    \]
    In the case (ii), we have
    $\left( \tau_i(a) \right)^n = E_{s_i}$, $\left( \tau_i(b) \right)^m = \left( \diag(h_i, \dots, h_i) \right)^{u_i} = E_{s_i}$, and 
    \begin{align*}    
        \tau_i(b) \tau_i(a) = 
        \begin{bNiceArray}{c|cw{c}{2.4cm}c}[margin]
            0 & \Block{3-3}<>{\diag(\alpha_i^r, \dots, \alpha_i^{r^{s_i - 1}})} & & \\
            \Vdots & & & \\
            0 & & & \\ \hline
            h_i \alpha_i & 0 & \Cdots & 0
        \end{bNiceArray} &= 
        \begin{bNiceArray}{c|cw{c}{2.4cm}c}[margin]
            0 & \Block{3-3}<>{\diag(\alpha_i^r, \dots, \alpha_i^{r^{s_i - 1}})} & & \\
            \Vdots & & & \\
            0 & & & \\ \hline
            \alpha_i^{r^{s_i}} h_i & 0 & \Cdots & 0
        \end{bNiceArray} = \\ &=
        \left( \tau_i(a) \right)^r \tau_i(b).
    \end{align*}
    Therefore, $\tau_i$, $i \in \intrange{1}{\omega}$, are indeed homomorphisms.

    Now, we prove that $\tau$ is injective. Given an arbitary element $P \in \FF_qG_{n,m,r}$, which can be represented as $P = \sum_{j=0}^{m-1}P_j(a)b^i$, where $P_j(x) \in \FF_q[x]_{< n}$, we have 
    \[
        \tau_i(P) = \sum_{j=0}^{m-1} P_{j}(\alpha_i) b^j
    \]
    in the case $s_i = 1$, and $\tau_i(P) =$
    \begin{equation*}
        = \sum_{z=0}^{u_i}
        \left[ \begin{array}{c|c|c|c|c}
            P_{zs_i}(\alpha_i) & 
            P_{1 + zs_i}(\alpha_i) &
            P_{2 + zs_i}(\alpha_i) & \cdots &
            P_{s_i-1 + zs_i}(\alpha_i) 
            \\ 
            P_{s_i-1 + zs_i}(\alpha_i^r) h_i & 
            P_{zs_i}(\alpha_i^r) &
            P_{1 + zs_i}(\alpha_i^r) & \cdots &
            P_{s_i-2 + zs_i}(\alpha_i^r)
            \\ 
            P_{s_i-2 + zs_i}(\alpha_i^{r^2}) h_i & 
            P_{s_i-1 + zs_i}(\alpha_i^{r^2}) h_i &
            P_{zs_i}(\alpha_i^{r^2}) & \cdots &
            P_{s_i-1 + zs_i}(\alpha_i^{r^2})
            \\ \hline
            \vdots & \vdots & \vdots & \ddots & \vdots 
            \\ \hline
            P_{1 + zs_i}(\alpha_i^{r^{s_i-1}}) h_i & 
            P_{2 + zs_i}(\alpha_i^{r^{s_i-1}}) h_i &
            P_{3 + zs_i}(\alpha_i^{r^{s_i-1}}) h_i & \cdots &
            P_{zs_i}(\alpha_i^{r^{s_i-1}})
        \end{array} \right] h_i^z
    \end{equation*}
    otherwise. Hence if $\tau(P) = 0$, then $P_j(\alpha_i^{r^l}) = 0$ for all $j \in \intrange{0}{m-1}$, $i \in \intrange{1}{\omega}$, $l \in \intrange{0}{s_i-1}$. This implies that all $P_j(x)$ are divisible by the polynomial $x^n-1 = \prod_{i=1}^{\omega}\prod_{l=0}^{s_i-1} f_i^{(r^l)}(x)$. Since $\deg P_j < n$, it follows that $P=0$, and therefore $\tau$ is injective.

    Finally, we have
    \begin{align*}
      \dim_{\FF_q}\left( \bigoplus_{i=1}^{\omega} \calA_i \right) &= 
      \sum_{i=1}^{\omega} \left( s_i^2 u_i \deg(f_i) \right) = m \sum_{i=1}^{\omega} \left( s_i \deg(f_i) \right) = \\ &=
      m \sum_{i=1}^{\omega} \sum_{l=0}^{s_i-1} \deg\left( f_i^{(r^l)} \right) = mn = \dim_{\FF_q}(\FF_qG_{n,m,r}).
    \end{align*}
    Therefore, $\tau$ is an $\FF_q$-algebra isomorphism.
\end{proof}

\begin{remark} \label{rem:decomp}
    If $\theta_i$ is the trivial homomorphism, i.e., when $\alpha_i = \alpha_i^{r^{s_i}}$, then $\FF_q[\alpha_i] \ast_{\theta_i} B_i$ equals the group algebra $\FF_q[\alpha_i] B_i$. Therefore, we have
    \begin{equation} \label{eq:ai}
        \calA_i = \begin{cases}
            \FF_q[\alpha_i] \ast_{\theta_i} B, & \text{$s_i = 1$ and $\alpha_i \neq \alpha_i^{r}$},  \\
            \Mat_{s_i}(\FF_q[\alpha_i] \ast_{\theta_i} B_i), & \text{$s_i > 1$ and $\alpha_i \neq \alpha_i^{r^{s_i}}$}, \\
            \FF_q[\alpha_i] B, & \text{$s_i = 1$ and $\alpha_i = \alpha_i^{r}$}, \\ 
            \Mat_{s_i}(\FF_q[\alpha_i] B_i), & \text{$s_i > 1$ and $\alpha_i = \alpha_i^{r^{s_i}}$}.
        \end{cases}         
    \end{equation} 
\end{remark}

\begin{remark}
    If $n$ is a divisor of $q-1$, then all factors of $x^n-1$ are of the form $f_i(x) = x - \alpha_i$, and hence all summands \eqref{eq:ai} are of the form $\FF_q[\alpha_i] B$ and $\Mat_{s_i}(\FF_q[\alpha_i] B_i)$. 
\end{remark}

For further study of metacyclic codes, only the decomposition given in \eqref{eq:decomp} and \eqref{eq:ai} will be leveraged. However, since skew group algebras can have a rather complex algebraic structure, it could be also useful to further refine this decomposition to eliminate skew group algebras. This can be done using
\begin{itemize}
    \item \Cref{prop:skew} if $\theta_i$ is a monomorphism. In this case, we have
    \[ 
        \FF_q[\alpha_i] \ast_{\theta_i} B_i \simeq \Mat_{u_i} (\FF_{q^{\deg(f_i)/u_i}}),
    \]
    and hence $\Mat_{s_i}(\FF_q[\alpha_i] \ast_{\theta_i} B_i) \simeq \Mat_{s_i \cdot u_i} (\FF_{q^{\deg(f_i)/u_i}}) = \Mat_m(\FF_{q^{\deg(f_i)/u_i}})$;
    \item  \Cref{col:skew2} if $B_i$ can be decomposed into an inner direct product of its subgroups $H_1, H_2$, such that $\theta_{i} |_{H_1}$ is a monomorphism and $\theta_{i} |_{H_2}$ is trivial;
    \item the evaluation isomorphism of \cite[Sect. 5]{Caruso2023} if $\gcd(u_i, q) = 1$ to decompose $\FF_q[\alpha_i] \ast_{\theta_i} B_i$ into a direct sum of matrix algebras over fields.
\end{itemize}
In particular, the following proposition refines \eqref{eq:decomp} in the case when $m$ is a prime.
\begin{proposition} \label{prop:mprime}
    Let $\gcd(q,n)=1$ and $m$ be prime. Let 
    \begin{align*}
        \Omega_1 &= \left\{ i \in \intrange{1}{\omega} \mid \alpha_i = \alpha_i^r \right\}, \\
        \Omega_2 &= \left\{ i \in \intrange{1}{\omega} \mid \alpha_i \neq \alpha_i^r, \; f_i(x) = f_i^{(r)}(x) \right\}, \\
        \Omega_3 &= \left\{ i \in \intrange{1}{\omega} \mid \alpha_i \neq \alpha_i^r, \; f_i(x) \neq f_i^{(r)}(x) \right\}.
    \end{align*}
    Then $\FF_qG_{n,m,r} \simeq$
    \begin{align} \label{eq:decomp_mprime}
        &\simeq 
        \left( \bigoplus_{i \in \Omega_1} \FF_q[\alpha_i] B \right) \oplus
        \left( \bigoplus_{i \in \Omega_2} \FF_q[\alpha_i] \ast_{\theta_i} B \right) \oplus
        \left( \bigoplus_{i \in \Omega_3} \Mat_m(\FF_q[\alpha_i]) \right) \simeq \\ \nonumber
        &\simeq
        \left( \bigoplus_{i \in \Omega_1} \FF_q[\alpha_i] B \right) \oplus
        \left( \bigoplus_{i \in \Omega_2} \Mat_m(
            \underbrace{
                \FF_q[\alpha_i + \alpha_i^r + \dots + \alpha_i^{r^{m-1}}]
            }_{\FF_{q^{\deg(f_i)/m}}}
        ) \right) \oplus \\ \nonumber &\oplus
        \left( \bigoplus_{i \in \Omega_3} \Mat_m(\FF_q[\alpha_i]) \right). \nonumber
    \end{align}
\end{proposition}
\begin{proof}
    If $m$ is prime, then each $B_i$ is either $B$ or $\{ e \}$. Hence, using Theorem \ref{theor:decomp} and \Cref{prop:skew}, we obtain the claim of the proposition.
\end{proof}

The following example illustrates that the decompositions of dihedral group algebras of \cite{Martinez2015} and generalized dihedral group algebras \cite{Gao2020} can be obtained as particular instances of \Cref{prop:mprime}.

\begin{example}
    Consider metacyclic groups of the form $G_{n,2,r}$. The isomorphism \eqref{eq:decomp}
    \[ 
        \tau: \FF_qG_{n,2,r} \to 
        \left( \bigoplus_{i \in \Omega_1} \FF_q[\alpha_i] B \right) \oplus
        \left( \bigoplus_{i \in \Omega_2} \FF_q[\alpha_i] \ast_{\theta_i} B \right) \oplus
        \left( \bigoplus_{i \in \Omega_3} \Mat_2(\FF_q[\alpha_i]) \right),
    \]
    is given by $\tau = \bigoplus_{i=1}^{\omega} \tau_i$, where
    \begin{itemize}
        \item for $i \in \Omega_1 \cup \Omega_2$:
        \[
            \tau_i(a) = \alpha_i, \quad \tau_i(b) = b;
        \]
        \item for $i \in \Omega_3$:
        \[
            \tau_i(a) = \begin{bmatrix}
                \alpha_i & 0 \\ 0 & \alpha_i^r
            \end{bmatrix}, \quad
            \tau_i(b) = \begin{bmatrix}
                0 & 1 \\ 1 & 0
            \end{bmatrix},
        \]
    \end{itemize}
    Let $i \in \Omega_2$. Recall that $|B| = 2$, $\theta_i(b)(\alpha_i) = \alpha_i^r$, and \Cref{prop:skew} implies that the map
    \begin{align*}
        \sigma_i: \; &\FF_q[\alpha_i] \ast_{\theta_i} B \to \End_{\FF_q[\alpha_i + \alpha_i^r]}(\FF_q[\alpha_i]), \\
        & P(\alpha_i) + Q(\alpha_i)b \mapsto P(\alpha_i) + Q(\alpha_i) \theta_{i}(b)
    \end{align*}
    is an isomorphism. One can easily note that the matrix representations of $\sigma_i(\alpha_i)$ and $\sigma_i(b)$ in the $\FF_q[\alpha_i + \alpha_i^r]$-basis $\{ 1, \alpha_i \}$ of $\FF_q[\alpha_i]$ are
    \[ 
        \begin{bmatrix}
            0 & -\alpha_i \alpha_i^r \\
            1 & \alpha_i + \alpha_i^r
        \end{bmatrix}, \quad
        \begin{bmatrix}
            1 & \alpha_i + \alpha_i^r \\
            0 & -1
        \end{bmatrix},
    \]
    respectively. Hence $\tilde{\sigma_i}: \FF_q[\alpha_i] \ast_{\theta_i} B \to \Mat_2(\FF_q[\alpha_i + \alpha_i^r])$ defined on $\FF_q$-generators of $\FF_q[\alpha_i]$ as follows:
    \[ 
        \tilde{\sigma_i}(\alpha_i) = \begin{bmatrix}
            0 & -\alpha_i \alpha_i^r \\
            1 & \alpha_i + \alpha_i^r
        \end{bmatrix}, \quad \quad
        \tilde{\sigma_i}(b) = 
        \begin{bmatrix}
            1 & \alpha_i + \alpha_i^r \\
            0 & -1
        \end{bmatrix},
    \]
    establishes a further isomorphism between $\FF_q[\alpha_i] \ast_{\theta_i} B$ and $\Mat_2(\FF_q[\alpha_i + \alpha_i^r])$ for $i \in \Omega_2$.
\end{example}

\section{Structure of metacyclic codes}
\label{sec:codes}
This section provides an algebraic description of metacyclic codes by leveraging results of the previous section. Additionally, in this section, a bound on the minimum distance is obtained, and it is shown that metacyclic codes can be viewed as generalized concatenated codes.

As is well-known, there exists a one-to-one correspondence between left ideals of $\Mat_l(R)$, where $R$ is a ring, and left $R$-submodules of $R^l$ (see \cite{Morita1958,Ferraz2021ideals}). Specifically, for any left $R$-submodule $L$ of $R^l$ (i.e., an additive subgroup of $R^l$ such that 
$\lambda l \in L$ for any $\lambda \in R$ and $l \in L$),
there is an associated left ideal of $\Mat_l(R)$ given by
\[ 
    \calI_l(L) = \left\{ 
        \begin{bmatrix}
            \text{---} & x^{(1)} & \text{---} \\ 
            \text{---} & x^{(2)} & \text{---} \\
            & \cdots & \\
            \text{---} & x^{(l)} & \text{---} 
        \end{bmatrix} \in \Mat_l(R) \;\; \Big| \;\; \forall i \in \intrange{1}{l} \;\;\;
        x^{(i)} = \left( x^{(i)}_1, \dots, x^{(i)}_l \right) \in L
    \right\}.
\]
Conversely, any left ideal is associated with a submodule consisting of all rows of matrices from the ideal.

Given that any left ideal in a direct sum of algebras is a direct sum of left ideals in the summands, \Cref{theor:decomp} immediately implies the following theorem which describes metacyclic codes.
\begin{theorem} \label{theor:codes}
    Let $\gcd(q,n) = 1$. Any left metacyclic code $C \subset \FF_qG_{n,m,r}$ can be uniquely desribded via its image under \eqref{eq:decomp}. Specifically,
    \begin{equation} \label{eq:code_decomp}
        \tau(C) = \bigoplus_{i=1}^{\omega} \calI_{s_i}(L_i),
    \end{equation}
    where each $L_i$ is a left $R_i$-submodule of $R_i^{s_i}$,
    \[
        R_i = \begin{cases}
            \FF_q[\alpha_i] B_i, & \alpha_i = \alpha_i^{r^{s_i}}, \\
            \FF_q[\alpha_i] \ast_{\theta_i} B_i, & \alpha_i \neq \alpha_i^{r^{s_i}}.
        \end{cases} 
    \]
\end{theorem}
It is worth mentioning that $L_i$ are also known as
\begin{itemize}
    \item \emph{cyclic codes} if $s_i = 1$ (and hence $u_i = m$ and $B_i = B$) and $\alpha_i=\alpha_i^{r}$ since in that case they are simply ideals of cyclic group algebras $\FF_q[\alpha_i] B$ ;
    \item \emph{skew cyclic codes} if $s_i = 1$ and $\alpha_i \neq \alpha_i^r$, with $L_i$ being left ideals of skew cyclic group algebras $\FF_q[\alpha_i] \ast_{\theta_i} B$; 
    \item \emph{linear codes over $\FF_q[\alpha_i]$} if $s_i = m$ (and hence $u_i = 1$ and $B_i = \{ e \}$), with each $L_i$ being a linear subspace of $\left( \FF_q[\alpha_i] \right)^m$;
    \item \emph{quasi-cyclic codes} if $1 < s_i < m$ and $\alpha_i \neq \alpha_i^{r^{s_i}}$, with $L_i$ being a $\FF_q[\alpha_i]B$-submodules of $\left( \FF_q[\alpha_i]B_i \right)^{s_i}$;
    \item \emph{skew quasi-cyclic codes} if $1 < s_i < m$ and $\alpha_i \neq \alpha_i^{r^{s_i}}$, with $L_i$ being a $\FF_q[\alpha_i] \ast_{\theta_i} B$-submodules of $\left( \FF_q[\alpha_i] \ast_{\theta_i} B_i \right)^{s_i}$
\end{itemize}
(see \Cref{rem:decomp}). Given that the algebraic description and properties of cyclic, quasi-cyclic, and skew cyclic codes are well-studied (see e.g. \cite{Ding2021, Guneri2021, Gluesing2021}), the characterization of $L_i$ in first four cases can be readily derived. Skew quasi-cyclic codes are studied in much less depth, however, their algebraic description can be derived via leveraging the decomposition of skew group algebras into direct sum of matrix algebras (see the previous section) and the one-to-one correspondence between left ideals of matrix algebras and left submodules.

Let $R$ be a ring. Given $x = (x^{(1)}, \dots, x^{(l)}) \in R^l$, by $\supp_{R^l}(x) = \{ j \in \intrange{1}{l} \mid x^{(j)} \neq 0 \}$ we denote the support of $x$. For a $R$-submodule $L$ of $R^l$, let
\[ \Supp_{R^l}(L) = \bigcup_{x \in L} \supp_{R^l}(x) \]
denote the union of supports of all $x \in L$. 

Let $\KK$ be a field, $G$ be a finite group, and let $G = \{ g_1, \dots, g_{|G|} \}$ be an enumeration of its elements. This enumeration induces a $\KK$-linear isomorphism $\operatorname{coord}_{\KK \ast_{\theta} G}: \KK \ast G \to \KK^{|G|}$ defined with respect to the enumeration by
\[ 
    \operatorname{coord}_{\KK \ast_{\theta} G}: \; 
        \sum_{g \in G} g \lambda_g
    \mapsto \left( \lambda_{g_1}, \dots, \lambda_{g_{|G|}} \right),
\]
where $\KK \ast G$ is a skew group algebra. Similarly, it is also possible to define a $\KK$-linear isomorphism $\operatorname{coord}_{\left( \KK \ast G \right)^l}: \left( \KK \ast G \right)^l \to \KK^{l \cdot |G|}$ by
\[
    \operatorname{coord}_{\left( \KK \ast G \right)^l}: \;
    \left( x^{(1)}, \dots, x^{(l)} \right) \mapsto 
    \left( 
        \operatorname{coord}_{\KK \ast_{\theta} G}(x^{(1)}) \; | \; \dots \; | \; \operatorname{coord}_{\KK \ast_{\theta} G}(x^{(l)}) 
    \right)
\]

Now, we are ready to establish a lower bound on the minimum distance of metacyclic codes.
\begin{theorem} \label{theor:mindist}
    Let $C \subset \FF_qG_{n,m,r}$ be a left metacyclic codes given by \eqref{eq:code_decomp}. Let $I = \left\{ i \in \intrange{1}{\omega} \mid L_i \neq \{ 0 \} \right\}$. For each $i \in I$, let
    \begin{itemize}
        \item $K_i = \Supp_{R_i^{s_i}}(L_i)$;
        \item $V_i \subset \FF_q A$ be a length-$n$ cyclic code defined by the following generator polynomial
        \[ 
            g(x) = (x^n-1) / \left( \prod_{j \in K_i} f_i^{(r^{j-1})}(x) \right), 
        \]
        i.e., $V_i = (\FF_q A)g(a)$;
        \item $d_i = d\left( \operatorname{coord}_{R_i^{s_i}}(L_i) \right)$ be the minimum distance of $L_i$ considered as a $\FF_q[\alpha_i]$-linear code.
    \end{itemize}
    Suppose that elements of $I = \{ i_1, \dots, i_{|I|} \}$ are enumerated such that
    $d_{i_1} \leq d_{i_2} \leq \dots \leq d_{i_{|I|}}$, then
    \begin{equation} \label{eq:min_dist}
        d(C) \geq \min_{1 \leq j \leq |I|} \left\{
            d_{i_j} \cdot d(V_{i_1} \dotplus \dots \dotplus V_{i_j})
        \right\}.
    \end{equation}    
\end{theorem}
\begin{proof}
    Let $P = \sum_{k=0}^{m-1} b^kP_k(a)$,  where $P_j(x) \in \FF_q[x]_{< n}$, be a non-zero codeword of $C$. It follows that there exists at least one $j \in \intrange{1}{|I|}$ such that $\tau_{i_j}(P) \neq 0$. Let $\tilde{j}$ be the largest such $j$.

    For $i = i_{\tilde{j}}$, we have $\tau_i(P) = \sum_{k=0}^{m-1} b^k P_k(\alpha_i)$ if $s_i = 1$ and $\tau_i(P) =$
    \begin{equation} \label{eq:matrix_tau}
            = \sum_{z=0}^{u_i} h_i^z
            \left[ \begin{array}{c|c|c|c|c}
                P_{zs_i}(\alpha_i) & 
                P_{1 + zs_i}(\alpha_i^r) &
                P_{2 + zs_i}(\alpha_i^{r^2}) & 
                \cdots &
                P_{s_i-1 + zs_i}(\alpha_i^{r^{s_i-1}}) 
                \\ 
                h_i P_{s_i-1 + zs_i}(\alpha_i) & 
                P_{zs_i}(\alpha_i^r) &
                P_{1 + zs_i}(\alpha_i^{r^2}) & 
                \cdots &
                P_{s_i-2 + zs_i}(\alpha_i^{r^{s_i-1}})
                \\ 
                h_i P_{s_i-2 + zs_i}(\alpha_i) & 
                h_i P_{s_i-1 + zs_i}(\alpha_i^{r}) &
                P_{zs_i}(\alpha_i^{r^2}) & 
                \cdots &
                P_{s_i-1 + zs_i}(\alpha_i^{r^{s_i-1}})
                \\ \hline
                \vdots & \vdots & \vdots & \ddots & \vdots 
                \\ \hline
                h_i P_{1 + zs_i}(\alpha_i) & 
                h_i P_{2 + zs_i}(\alpha_i^{r}) &
                h_i P_{3 + zs_i}(\alpha_i^{r^2}) & 
                \cdots &
                P_{zs_i}(\alpha_i^{r^{s_i-1}})
            \end{array} \right], 
    \end{equation}
    otherwise. Given the definition of $\calI_{s_i}(L_i)$, $\tau_i(P) \neq 0$ implies that there exist at least $d_i$ indices $k_1, k_2, \dots, k_{d_i}$ such that $P_{k_1}(a), \dots, P_{k_{d_i}}(a)$ are non-zero. Moreover, one can easily note that 
    \[ 
        P_{k_1}(a), \dots, P_{k_{d_i}}(a) \in V_{i_1} \dotplus \dots \dotplus V_{i_{\tilde{j}}}
    \] 
    due to the definitions of $\tilde{j}$, $K_i$, and $V_i$. Hence using $\wt{P} = \sum_{k=0}^{m-1} \wt{P_k}$, we obtain 
    \[ 
        \wt{P} \geq d_{i_{\tilde{j}}} \cdot d(V_{i_1} \dotplus \dots \dotplus V_{i_{\tilde{j}}}),
    \]
    which implies \eqref{eq:min_dist}.    
\end{proof}

\paragraph{Concatenated Structure of Metacyclic Codes}
In the following, we show that metacyclic codes can also be viewed as generalized concatenated (GC) codes (\cite{Blokh1974, Zyablov1999}). 

Generalized concatenated codes are an effective construction for building long codes from shorter ones.  Given:
\begin{itemize}
    \item $l$ outer $\FF_{q^{m_i}}$-linear codes $C_1 \subset \FF_{q^{m_1}}, \dots, C_l \subset \FF_{q^{m_l}}$ of length $m$;
    \item $n$ inner $\FF_q$-linear codes $Z_1, \dots, Z_n$ of length $n$ and dimension $\sum_{i=1}^{l} m_i$;
    \item $n$ $\FF_q$-linear encoding maps $\psi_j: \bigoplus_{i=1}^l \FF_{q^{m_i}} \to Z_j \subset \FF_q^n,$
\end{itemize}
the \emph{generalized concatenated code} $(C_1, \dots, C_l) \Box (Z_1, \dots, Z_n)$ is defined by
\[
  \left\{ 
    \begin{bmatrix}
        \text{---} & \psi_1(c_{1,1}, c_{1,2}, \dots, c_{1,l}) & \text{---} \\
        \text{---} & \psi_2(c_{2,1}, c_{2,2}, \dots, c_{2,l}) & \text{---} \\
        & \vdots & \\
        \text{---} & \psi_n(c_{m,1}, c_{m,2}, \dots, c_{m,l}) & \text{---}
    \end{bmatrix} \in \FF_q^{m \times n} \; \Bigg| \; \left( c_{1,i}, c_{2,i}, \dots, c_{m,i} \right) \in C_i 
  \right\}.  
\]
Simply put, the encoding of GC codes can be performed in two steps: first, we construct a matrix whose columns are codewords of outer codes $C_1, \dots, C_l$, and then encode each row of this matrix using inner codes. Note that often the inner codes (as well as corresponding encoding maps) are usually chosen to coincide with each other, however, generally this is not required. 

In this section, we will consider a slightly different presentation of this construction by assuming the outer codes are $G$-codes and the codewords of the resulting GC codes are elements of $(\FF_q G)^m$ instead of $(m \times n)$-matrices.

Consider a metacyclic code $C$. Below, we rely on the notations of \Cref{theor:mindist} and its proof. One can easily note that by performing the following steps to matrix \eqref{eq:matrix_tau}:
\begin{enumerate}[(i)]
    \item apply $\operatorname{coord}_{R_i^{s_i}}$ to each row and transpose the resulting matrix;
    \item rearrange items in each column to obtain the matrix having values of $P_k(x)$ in its $k$-th row for all $k \in \intrange{1}{m}$,
\end{enumerate}
we obtain the matrix, each column of which is a codeword of a code permutationally equivalent to $\operatorname{coord}_{R_i^{s_i}}(L_i)$, with each $k$-th row being the evaluation vector of the polynomial $P_k(x)$ at points $\{ \alpha_i, \alpha_i^{r}, \dots, \alpha_i^{r^{s_i-1}} \}$. For example, 
\[
    \begin{bmatrix}
        P_0(\alpha_i) + h_i P_2(\alpha_i) & P_1(\alpha_i^r) + h P_3(\alpha_i^r) \\
        P_3(\alpha_i) + h_i P_1(\alpha_i) & P_0(\alpha_i^r) + h_i P_2(\alpha_i^r)
    \end{bmatrix} \mapsto
    \begin{bmatrix}
        P_0(\alpha_i) & P_0(\alpha_i^r) \\
        P_1(\alpha_i^r) & P_1(\alpha_i) \\
        P_2(\alpha_i) & P_2(\alpha_i^r) \\
        P_3(\alpha_i^r) & P_3(\alpha_i)
    \end{bmatrix},
\]
\[
    \begin{bmatrix}
        P_0(\alpha_i) & P_1(\alpha_i^r) & P_2(\alpha_i^{r^2}) \\
        P_2(\alpha_i) & P_0(\alpha_i^r) & P_1(\alpha_i^{r^2}) \\
        P_1(\alpha_i) & P_2(\alpha_i^r) & P_0(\alpha_i^{r^2})
    \end{bmatrix} \mapsto
    \begin{bmatrix}
        P_0(\alpha_i) & P_0(\alpha_i^r) & P_0(\alpha_i^{r^2}) \\
        P_1(\alpha_i^r) & P_1(\alpha_i^{r^2}) & P_1(\alpha_i) \\
        P_2(\alpha_i^{r^2}) & P_2(\alpha_i) & P_2(\alpha_i^r)
    \end{bmatrix}.
\]

For different $i \in I$, the resulting matrices after steps (i) and (ii) can be concatenated side by side. It follows that the encoding of metacyclic codes can be performed in the same two steps as encoding of GC codes: 
\begin{itemize}
    \item first, we obtain a matrix consisting of codewords of some codes permutationally equivalent to $\operatorname{coord}_{R_i^{s_i}}(L_i)$, $i \in I$;
    \item second, we recover each $P_k(x)$ from its evaluations (in this step we obtain $P_k(a)$ as codewords of $V_{i_1} \dotplus \dots \dotplus V_{i_{|I|}}$).
\end{itemize}
Therefore, metacyclic codes can be indeed viewed as generalized concatenated codes, with outer codes being skew quasi-cyclic codes (in the most general case), and inner codes being the cyclic code $V = V_{i_{1}} \dotplus \dots \dotplus V_{i_{|I|}}$ of length $n$. In fact, the distance bound obtained in \Cref{theor:mindist} coincides with the minimum distance bound of metacyclic codes viewed as GC codes. Additionally, the GC structure of metacyclic codes also allows for using decoding methods for GC codes for decoding metacyclic codes. 

\begin{remark}
In \cite{Puchinger2017}, S.~Puchinger, S.~M\"uelich, K.~Ishak, and M.~Bossert described an attack that, under certain conditions, allows for partially recovering the secret permutation of McEliece-type cryptosystems based on generalized concatenated (GC) codes. It was shown in \cite{Puchinger2017} that this enables a significant reduction in the complexity of message-recovery attacks. Consequently, the generalized concatenated structure of metacyclic codes implies that many  instances of cryptosystems based on these codes can be effectively broken using the PMIB-attack.

Note that a sufficient condition for this attack to work is the existence of a large number of codewords in $C^{\perp}$ of weight less than $\min_{i \in I}\left\{ d\big( \operatorname{coord}_{R_i^{s_i}}(L_i)^\perp \big), 2d\left( V^{\perp} \right) \right\}$.
\end{remark}

\section{Induced codes}
\label{sec:induced}
Let \(G\) be a group and \(H\) be its subgroup of index \(|G:H|\). Given a left \(H\)-code \(C \subset \FF_qH\), it is possible to obtain the following code
\[ 
    C^G = (\FF_qG) \otimes_{\FF_qH} C = (\FF_qG) \cdot C,
\] 
referred to as the \emph{\(G\)-code induced by an \(H\)-code \(C\)} (or simply an \emph{induced code}) \cite{Zimmermann1994}. 
Let \(T_L(G, H) = \{ g_1, \dots, g_{|G:H|} \}\) be a left transversal of \(H\), i.e., 
\[ G = \bigsqcup_{g \in T_L(G, H)} gH. \]
Since any element of \(\FF_qG\) can be uniquely represented as \(\sum_{g \in T_L(G,H)} g u_g\), where \(u_g \in \FF_qH\), it follows that 
\[ \FF_qG = g_1(\FF_qH) \dotplus g_2(\FF_qH) \dotplus \dots \dotplus g_{|G:H|}(\FF_qH), \]
and hence
\begin{equation} \label{d:eq:CG}
    C^G = g_1C \dotplus g_2 C \dotplus \dots \dotplus g_{|G:H|} C. 
\end{equation} 
Therefore, \(C^G\) can be considered as the repeated \(|G:H|\) times code \(C\), with each \(g_i C\) protecting symbols of \(\FF_qG\) indexed by the coset \(g_iH\). In particular, that means that if \(C\) is a \([n, k, d]\)-code, then \(C^G\) is a \( \left[n |G:H|, k |G:H|, d \right]\)-code, and if \(\boldB(C)\) is a basis of \(C\), then 
\[ T(G,H) \cdot \boldB(C)= \left\{ g_i \cdot \mathbf{b} \mid g_i \in T_L(G,H), \; \mathbf{b} \in \boldB(C)  \right\} \]
is a basis of \(C^G\) (see also \cite{Deundyak2016}).

Since it is often easier to study codes from subgroups (e.g., cyclic codes) than from the group itself, induced codes could be a useful tool for studying the properties of group codes. 

\begin{definition} \label{d:def:exterior}
    Let $G$ be a group and $H$ be its subgroup. Given a $G$-code $C \subset \FF_qG$, the intersection $\Ext_H(C)$ of all $G$-codes induced by $H$-codes and containing $C$ is called the \emph{exterior induced $H$-code of $C$}. In other words, $\Ext_H(C)$ is the smallest induced code containing $C$.
\end{definition}

\begin{proposition} \label{d:prop:exterior}
    Let the projection $\operatorname{pr}_{H}: \FF_qG \to \FF_qH$ be given by
    \[ 
        \operatorname{pr}_{H}\left( \sum_{g \in G} \lambda_g g \right) = \sum_{g \in H} \lambda_g g,
    \]
    and let $C \subset \FF_qG$ be a $G$-code. Then $\operatorname{pr}_{H}(C)$ is a $H$-code and $\Ext_{H}(C) = \left( \operatorname{pr}_{H}(C) \right)^G$.     
\end{proposition}
\begin{proof}
    One can easily note that $\operatorname{pr}_{H}$ is a surjective $\FF_q$-linear map such that
    \[ 
        \operatorname{pr}_{H}(h u) = h  \operatorname{pr}_{H}(u).
    \]
    for any $h \in H$ and $u \in \FF_qG$. It follows that the image of any left ideal of $\FF_qG$ under $\operatorname{pr}_{H}$ is a left ideal of $\FF_qH$, so $\operatorname{pr}_{H}(C)$ is indeed an $H$-code.

    Recall that any element $u \in \FF_qG$ can be uniquely represented as $\sum_{g \in T_L(G, H)} g u_g$, where $u_g \in \FF_qH$. Since $C$ is a left ideal, it follows that for any $u \in C$ we have
    \[  
        \big\{ u_g = \operatorname{pr}_H(\underbrace{g^{-1} u }_{\in C}) \mid g \in T_L(G, H) \big\} \subset \operatorname{pr}_{H}(C).
    \]
    Therefore, by \eqref{d:eq:CG}, $u \in \left( \operatorname{pr}_{H}(C) \right)^G$, and hence $C \subset \left( \operatorname{pr}_{H}(C) \right)^G$. Consequentially, $\Ext_H(G) \subset \left( \operatorname{pr}_{H}(C) \right)^G$.

    Now, let $I$ be an arbitrary $H$-code such that $C \subset I^G$. Using \eqref{d:eq:CG}, we infer $\operatorname{pr}_{H}(I^G) = I$, and hence 
    \[
        C \subset I^G \implies \operatorname{pr}_{H}(C) \subset \underbrace{\operatorname{pr}_{H}(I^G)}_{ = I} \implies 
        \left( \operatorname{pr}_{H}(C) \right)^G \subset I^G.
    \]
    By choosing $I^G$ to be $\Ext_H(G)$, we obtain
   $\left( \operatorname{pr}_{H}(C) \right)^G \subset \Ext_H(G)$.
\end{proof}

\begin{corollary}  \label{d:cor:6}
    By \Cref{d:prop:exterior} and \eqref{d:eq:CG}, for any codeword $u$ of a $G$-code $C$ we have
    \begin{equation} \label{d:eq:cor5}
        u = \sum_{g \in T_L(G, H)} g u_g, \quad u_g \in \operatorname{pr}_{H}(C), 
    \end{equation}         
    and hence $d(C) \geq d(\operatorname{pr}_H(C))$.
\end{corollary}

\begin{definition} \label{d:def:interior}
    Let $G$ be a group, and $H$ be its subgroup. Given a $G$-code $C \subset \FF_qG$, the sum $\operatorname{Int}_H(C)$ of all $G$-codes contained in $C$ and induced by $H$-codes is called the \emph{interior induced $H$-code of $C$}. In other words, $\operatorname{Int}_H(C)$ is the largest induced subcode of $C$.
\end{definition}

\begin{proposition} \label{d:prop:interior}
    Let $C \subset \FF_qG$ be a $G$-code and let $C|_H = \operatorname{pr}_H(C \cap \FF_qH)$. Then $C|_H$ is a $H$-code and $\operatorname{Int}_H(C) = \left( C|_H \right)^G$. Furthermore, $d(C|_{\mathcal{H}}) > d(C)$.
\end{proposition}
\begin{proof}
    The first claim is obvious. Now, we'll show that $\operatorname{Int}_H(C) = \left( C|_H \right)^G$. Indeed, consider an $H$-code $I$ such that $I^G \subset C$. It follows that
    \[
        \operatorname{pr}_H(I^G \cap \FF_qH) \subset \operatorname{pr}_H(C \cap \FF_qH).
    \]    
    Using \eqref{d:eq:CG}, we infer
    \( 
        \operatorname{pr}_H(I^G \cap \FF_qH) = I    
    \),
    and hence $I \subset \left( C|_H \right)^G$. Since $I$ can be chosen arbitrary, we obtain $\operatorname{Int}_{H}(C) \subset \left( C|_H \right)^G$ (see \Cref{d:def:interior}). On the other hand, with $\left( C|_H \right)^G$ being an induced code contained in $C$, we have $\left( C|_H \right)^G \subset \operatorname{Int}_{H}(C)$ by \Cref{d:def:interior}. 
\end{proof}

Below, the rest of this section provides an explicit decomposition of codes induced by $A$-codes, introduces a class of codes obtained by intersecting induced codes, and derives another lower bound on the minimum distance by leveraging them.

\begin{proposition}
    Let $g(x)$ be a divisor of $x^n-1$, and let $C = (\FF_q A) g(a)$ be the cyclic code generated by $g$. Then $C^{G_{n,m,r}}$ has the following decomposition
    \begin{equation*}
        \tau(C) = \bigoplus_{i=1}^{\omega} \calI_{s_i}(L_i),
    \end{equation*}
    where 
    \[ 
        L_i = \left\{ \left( x^{(0)}, \dots, x^{(s_i-1)} \right) \mid x^{(j)} = 0 \text{ for all $j$ s.t. $g(\alpha_i^{r^j}) = 0$}  \right\}.
    \]
\end{proposition}
\begin{proof}
    Directly follows from \Cref{theor:codes} and \eqref{eq:matrix_tau}.
\end{proof}

\paragraph{Intersection of induced codes} 
While pure induced codes have rather poor parameters, they can be leveraged for building more powerful codes. For example, a class of such codes can be obtained by intersecting the induced codes from distinct subgroups. The following theorem provides a lower bound on their minimum distance and dimension.
\begin{theorem}
    Let $G$ be a group, and let $H_1, H_2$ be its subgroups such that $G = H_1 H_2$ and $H_1 \cap H_2 = \{ e \}$. Let $C_1 \subset \FF_qH_1$ be a $H_1$-code, and let $C_2 \subset \FF_qH_2$ be a $H_2$-code. Let $C = C_1^G \cap C_2^G$. Then
    \[
        d(C) \geq d(C_1) \cdot d(C_2)
    \]
    and $\dim(C) \geq |H_1| \cdot \dim(C_2) + |H_2| \cdot \dim(C_1) - |G|$.
\end{theorem}
    
\begin{proof}
    Let $x \in \FF_qG$ be a non-zero codeword of $C$. Since $x \in C_1^G$, using \eqref{d:eq:CG}, we obtain that there exists $g \in G$ such that
    \[ | \supp(x) \cap gH_1 | \geq d(C_1). \]
    Let $g_1, \dots, g_{d(C_1)} \in \supp(x) \cap gH_1$. One can easily note that $g_1H_2, \dots, g_{d(C_1)}H_2$ are distinct cosets, since the contrary implies that $|g_i H_1 \cap g_i H_2| \geq 2$ for some $i$, which is possible if and only if $|H_1 \cap H_2| \neq 1$.
    
    Therefore, since $x \in C_2^G$, it follows that 
    \begin{align*} 
        |\supp(x) \cap g_1 H_2| &\geq d(C_2), \\
        |\supp(x) \cap g_2 H_2| &\geq d(C_2), \\ 
        &\dots \\
        |\supp(x) \cap g_{d(C_1)} H_2| &\geq d(C_2),
    \end{align*}
    and hence $d(C) \geq d(C_1) \cdot d(C_2)$.

    The lower bound on the dimension of $C$ directly follows from the fact that
    \[ \dim(C) \geq |G| - \left(|G| - \dim(C_1^G)\right) - \left(|G| - \dim(C_2^G)\right), \]
    which simplifies to $\dim(C) \geq |H_1| \cdot \dim(C_2) + |H_2| \cdot \dim(C_1) - |G|$.
\end{proof}
Note that, from the proof of the theorem, it follows that intersections of induced codes can be viewed as product-like generalized LDPC codes (see \cite{Lentmaier2010}). This implies that it is possible to leverage decoding techniques of GLDPC codes for decoding metacyclic codes.

The following corollary provides another lower bound on the minimum distance of metacyclic codes by leveraging exterior induced codes (see \Cref{d:prop:exterior}).
\begin{corollary}
    Let $C \subset G_{n,m,r}$ be a metacyclic code. Then
    \[ 
        d(C) \geq d \left( \operatorname{pr}_{A}(C) \right) \cdot d \left( \operatorname{pr}_{B}(C) \right)
    \]
\end{corollary}

\section{Conclusion}
\label{sec:concl}
In this paper, an explicit decomposition of split metacyclic group algebras is provided assuming the only restriction $\gcd(q,n)=1$. This decomposition has been further employed to obtain an algebraic description of metacyclic codes. Furthermore, the obtained structure has enabled the discovery of the concatenated structure of metacyclic codes and the development of a partial key-recovery attack against cryptosystems based on \emph{certain} metacyclic codes. Additionally, the paper provides results on induced codes, as well as estimates of the main parameters of metacyclic codes.

Further research directions may include improving the estimates of parameters obtained, finding efficient classes of metacyclic codes and decoding algorithms for them, and studying their applications, including cryptographic ones.

\backmatter

\bmhead{Acknowledgements}
The author would like to thank Yury Kosolapov for helpful comments and discussions.

\bmhead{Disclosure of Interests}
The author has no competing interests to declare that are relevant to the content of this article.


\bibliography{refs}


\begin{thebibliography}{38}
\ifx \bisbn   \undefined \def \bisbn  #1{ISBN #1}\fi
\ifx \binits  \undefined \def \binits#1{#1}\fi
\ifx \bauthor  \undefined \def \bauthor#1{#1}\fi
\ifx \batitle  \undefined \def \batitle#1{#1}\fi
\ifx \bjtitle  \undefined \def \bjtitle#1{#1}\fi
\ifx \bvolume  \undefined \def \bvolume#1{\textbf{#1}}\fi
\ifx \byear  \undefined \def \byear#1{#1}\fi
\ifx \bissue  \undefined \def \bissue#1{#1}\fi
\ifx \bfpage  \undefined \def \bfpage#1{#1}\fi
\ifx \blpage  \undefined \def \blpage #1{#1}\fi
\ifx \burl  \undefined \def \burl#1{\textsf{#1}}\fi
\ifx \doiurl  \undefined \def \doiurl#1{\url{https://doi.org/#1}}\fi
\ifx \betal  \undefined \def \betal{\textit{et al.}}\fi
\ifx \binstitute  \undefined \def \binstitute#1{#1}\fi
\ifx \binstitutionaled  \undefined \def \binstitutionaled#1{#1}\fi
\ifx \bctitle  \undefined \def \bctitle#1{#1}\fi
\ifx \beditor  \undefined \def \beditor#1{#1}\fi
\ifx \bpublisher  \undefined \def \bpublisher#1{#1}\fi
\ifx \bbtitle  \undefined \def \bbtitle#1{#1}\fi
\ifx \bedition  \undefined \def \bedition#1{#1}\fi
\ifx \bseriesno  \undefined \def \bseriesno#1{#1}\fi
\ifx \blocation  \undefined \def \blocation#1{#1}\fi
\ifx \bsertitle  \undefined \def \bsertitle#1{#1}\fi
\ifx \bsnm \undefined \def \bsnm#1{#1}\fi
\ifx \bsuffix \undefined \def \bsuffix#1{#1}\fi
\ifx \bparticle \undefined \def \bparticle#1{#1}\fi
\ifx \barticle \undefined \def \barticle#1{#1}\fi
\bibcommenthead
\ifx \bconfdate \undefined \def \bconfdate #1{#1}\fi
\ifx \botherref \undefined \def \botherref #1{#1}\fi
\ifx \url \undefined \def \url#1{\textsf{#1}}\fi
\ifx \bchapter \undefined \def \bchapter#1{#1}\fi
\ifx \bbook \undefined \def \bbook#1{#1}\fi
\ifx \bcomment \undefined \def \bcomment#1{#1}\fi
\ifx \oauthor \undefined \def \oauthor#1{#1}\fi
\ifx \citeauthoryear \undefined \def \citeauthoryear#1{#1}\fi
\ifx \endbibitem  \undefined \def \endbibitem {}\fi
\ifx \bconflocation  \undefined \def \bconflocation#1{#1}\fi
\ifx \arxivurl  \undefined \def \arxivurl#1{\textsf{#1}}\fi
\csname PreBibitemsHook\endcsname

\bibitem[\protect\citeauthoryear{Milies and Sehgal}{2002}]{Miles2002}
\begin{bbook}
\bauthor{\bsnm{Milies}, \binits{C.P.}},
\bauthor{\bsnm{Sehgal}, \binits{S.K.}}:
\bbtitle{An Introduction to Group Rings}.
\bsertitle{Algebra and Applications}.
\bpublisher{Springer}, \blocation{???}
(\byear{2002})
\end{bbook}
\endbibitem

\bibitem[\protect\citeauthoryear{Willems}{2021}]{Willems2021}
\begin{bchapter}
\bauthor{\bsnm{Willems}, \binits{W.}}:
\bctitle{Codes in group algebras}.
In: \bbtitle{Concise Encyclopedia of Coding Theory},
pp. \bfpage{363}--\blpage{384}.
\bpublisher{Chapman and Hall/CRC}, \blocation{???}
(\byear{2021})
\end{bchapter}
\endbibitem

\bibitem[\protect\citeauthoryear{Berman}{1969}]{Berman1967}
\begin{barticle}
\bauthor{\bsnm{Berman}, \binits{S.D.}}:
\batitle{On the theory of group codes}.
\bjtitle{Cybernetics}
\bvolume{3}(\bissue{1}),
\bfpage{25}--\blpage{31}
(\byear{1969})
\doiurl{10.1007/bf01072842}
\end{barticle}
\endbibitem

\bibitem[\protect\citeauthoryear{MacWilliams}{1970}]{MacWilliams1970}
\begin{barticle}
\bauthor{\bsnm{MacWilliams}, \binits{F.J.}}:
\batitle{Binary codes which are ideals in the group algebra of an abelian group}.
\bjtitle{Bell System Technical Journal}
\bvolume{49}(\bissue{6}),
\bfpage{987}--\blpage{1011}
(\byear{1970})
\doiurl{10.1002/j.1538-7305.1970.tb01812.x}
\end{barticle}
\endbibitem

\bibitem[\protect\citeauthoryear{Deundyak and Kosolapov}{2015}]{Deundyak2015}
\begin{barticle}
\bauthor{\bsnm{Deundyak}, \binits{V.M.}},
\bauthor{\bsnm{Kosolapov}, \binits{Y.V.}}:
\batitle{Algorithms for majority decoding of group codes}.
\bjtitle{Modeling and Analysis of Information Systems}
\bvolume{22}(\bissue{4}),
\bfpage{464}
(\byear{2015})
\doiurl{10.18255/1818-1015-2015-4-464-482}
\end{barticle}
\endbibitem

\bibitem[\protect\citeauthoryear{Deundyak and Kosolapov}{2016}]{Deundyak2016}
\begin{barticle}
\bauthor{\bsnm{Deundyak}, \binits{V.M.}},
\bauthor{\bsnm{Kosolapov}, \binits{Y.V.}}:
\batitle{Cryptosystem based on induced group codes}.
\bjtitle{Modeling and Analysis of Information Systems}
\bvolume{23}(\bissue{2}),
\bfpage{137}--\blpage{152}
(\byear{2016})
\doiurl{10.18255/1818-1015-2016-2-137-152} .
\bcomment{(in Russian)}
\end{barticle}
\endbibitem

\bibitem[\protect\citeauthoryear{Vedenev and Deundyak}{2018}]{Vedenev2018}
\begin{barticle}
\bauthor{\bsnm{Vedenev}, \binits{K.V.}},
\bauthor{\bsnm{Deundyak}, \binits{V.M.}}:
\batitle{Codes in dihedral group algebra}.
\bjtitle{Modeling and Analysis of Information Systems}
\bvolume{25}(\bissue{2}),
\bfpage{232}--\blpage{245}
(\byear{2018})
\doiurl{10.18255/1818-1015-2018-2-232-245}
\end{barticle}
\endbibitem

\bibitem[\protect\citeauthoryear{Vedenev and Deundyak}{2019}]{Vedenev2019}
\begin{barticle}
\bauthor{\bsnm{Vedenev}, \binits{K.V.}},
\bauthor{\bsnm{Deundyak}, \binits{V.M.}}:
\batitle{Codes in a dihedral group algebra}.
\bjtitle{Automatic Control and Computer Sciences}
\bvolume{53}(\bissue{7}),
\bfpage{745}--\blpage{754}
(\byear{2019})
\doiurl{10.3103/s0146411619070198}
\end{barticle}
\endbibitem

\bibitem[\protect\citeauthoryear{Vedenev and Deundyak}{2020a}]{Vedenev2020}
\begin{barticle}
\bauthor{\bsnm{Vedenev}, \binits{K.V.}},
\bauthor{\bsnm{Deundyak}, \binits{V.M.}}:
\batitle{Relationship between codes and idempotents in a dihedral group algebra}.
\bjtitle{Mathematical Notes}
\bvolume{107}(\bissue{1–2}),
\bfpage{201}--\blpage{216}
(\byear{2020})
\doiurl{10.1134/s0001434620010204}
\end{barticle}
\endbibitem

\bibitem[\protect\citeauthoryear{Vedenev and Deundyak}{2020b}]{Vedenev2020-2}
\begin{botherref}
\oauthor{\bsnm{Vedenev}, \binits{K.V.}},
\oauthor{\bsnm{Deundyak}, \binits{V.M.}}:
Some properties of dihedral group codes
(2020)
\end{botherref}
\endbibitem

\bibitem[\protect\citeauthoryear{Sabin}{1994}]{Sabin1994}
\begin{barticle}
\bauthor{\bsnm{Sabin}, \binits{R.E.}}:
\batitle{On row-cyclic codes with algebraic structure}.
\bjtitle{Designs, Codes and Cryptography}
\bvolume{4}(\bissue{2}),
\bfpage{145}--\blpage{155}
(\byear{1994})
\doiurl{10.1007/bf01578868}
\end{barticle}
\endbibitem

\bibitem[\protect\citeauthoryear{Sabin and Lomonaco}{1995}]{Sabin1995}
\begin{barticle}
\bauthor{\bsnm{Sabin}, \binits{R.E.}},
\bauthor{\bsnm{Lomonaco}, \binits{S.J.}}:
\batitle{Metacyclic error-correcting codes}.
\bjtitle{Applicable Algebra in Engineering, Communication and Computing}
\bvolume{6}(\bissue{3}),
\bfpage{191}--\blpage{210}
(\byear{1995})
\doiurl{10.1007/bf01195337}
\end{barticle}
\endbibitem

\bibitem[\protect\citeauthoryear{Assuena and Milies}{2016}]{Assuena2016}
\begin{barticle}
\bauthor{\bsnm{Assuena}, \binits{S.}},
\bauthor{\bsnm{Milies}, \binits{C.P.}}:
\batitle{Group algebras of metacyclic groups over finite fields}.
\bjtitle{São Paulo Journal of Mathematical Sciences}
\bvolume{11}(\bissue{1}),
\bfpage{46}--\blpage{52}
(\byear{2016})
\doiurl{10.1007/s40863-016-0043-7}
\end{barticle}
\endbibitem

\bibitem[\protect\citeauthoryear{Assuena and Milies}{2019}]{assuena2019}
\begin{barticle}
\bauthor{\bsnm{Assuena}, \binits{S.}},
\bauthor{\bsnm{Milies}, \binits{C.P.}}:
\batitle{Good codes from metacyclic groups}.
\bjtitle{Contemp. Math}
\bvolume{727},
\bfpage{39}--\blpage{49}
(\byear{2019})
\end{barticle}
\endbibitem

\bibitem[\protect\citeauthoryear{Assuena}{2020}]{Assuena2020}
\begin{botherref}
\oauthor{\bsnm{Assuena}, \binits{S.}}:
Good codes from metacyclic groups ii.
Journal of Algebra and Its Applications
\textbf{21}(02)
(2020)
\doiurl{10.1142/s0219498822500402}
\end{botherref}
\endbibitem

\bibitem[\protect\citeauthoryear{Broche and Del~R{\'\i}o}{2007}]{Broche2007}
\begin{barticle}
\bauthor{\bsnm{Broche}, \binits{O.}},
\bauthor{\bsnm{Del~R{\'\i}o}, \binits{{\'A}.}}:
\batitle{Wedderburn decomposition of finite group algebras}.
\bjtitle{Finite Fields and Their Applications}
\bvolume{13}(\bissue{1}),
\bfpage{71}--\blpage{79}
(\byear{2007})
\doiurl{10.1016/j.ffa.2005.08.002}
\end{barticle}
\endbibitem

\bibitem[\protect\citeauthoryear{Olteanu and Van~Gelder}{2014}]{Olteanu2014}
\begin{barticle}
\bauthor{\bsnm{Olteanu}, \binits{G.}},
\bauthor{\bsnm{Van~Gelder}, \binits{I.}}:
\batitle{Construction of minimal non-abelian left group codes}.
\bjtitle{Designs, Codes and Cryptography}
\bvolume{75}(\bissue{3}),
\bfpage{359}--\blpage{373}
(\byear{2014})
\doiurl{10.1007/s10623-014-9922-z}
\end{barticle}
\endbibitem

\bibitem[\protect\citeauthoryear{Brochero~Martinez}{2015}]{Martinez2015}
\begin{barticle}
\bauthor{\bsnm{Brochero~Martinez}, \binits{F.E.}}:
\batitle{Structure of finite dihedral group algebra}.
\bjtitle{Finite Fields and Their Applications}
\bvolume{35},
\bfpage{204}--\blpage{214}
(\byear{2015})
\doiurl{10.1016/j.ffa.2015.05.002}
\end{barticle}
\endbibitem

\bibitem[\protect\citeauthoryear{Vedenev and Kosolapov}{2021}]{Vedenev2021}
\begin{bchapter}
\bauthor{\bsnm{Vedenev}, \binits{K.}},
\bauthor{\bsnm{Kosolapov}, \binits{Y.}}:
\bctitle{On squares of dihedral codes}.
In: \bbtitle{2021 XVII International Symposium" Problems of Redundancy in Information and Control Systems"(REDUNDANCY)},
pp. \bfpage{55}--\blpage{60}
(\byear{2021}).
\bcomment{IEEE}
\end{bchapter}
\endbibitem

\bibitem[\protect\citeauthoryear{Gao et~al.}{2020}]{Gao2020}
\begin{barticle}
\bauthor{\bsnm{Gao}, \binits{Y.}},
\bauthor{\bsnm{Yue}, \binits{Q.}},
\bauthor{\bsnm{Wu}, \binits{Y.}}:
\batitle{Lcd codes and self-orthogonal codes in generalized dihedral group algebras}.
\bjtitle{Designs, Codes and Cryptography}
\bvolume{88}(\bissue{11}),
\bfpage{2275}--\blpage{2287}
(\byear{2020})
\doiurl{10.1007/s10623-020-00778-z}
\end{barticle}
\endbibitem

\bibitem[\protect\citeauthoryear{Cao et~al.}{2016a}]{Cao2016}
\begin{barticle}
\bauthor{\bsnm{Cao}, \binits{Y.}},
\bauthor{\bsnm{Cao}, \binits{Y.}},
\bauthor{\bsnm{Fu}, \binits{F.-W.}}:
\batitle{Concatenated structure of left dihedral codes}.
\bjtitle{Finite Fields and Their Applications}
\bvolume{38},
\bfpage{93}--\blpage{115}
(\byear{2016})
\doiurl{10.1016/j.ffa.2016.01.001}
\end{barticle}
\endbibitem

\bibitem[\protect\citeauthoryear{Cao et~al.}{2016b}]{Cao2016-2}
\begin{barticle}
\bauthor{\bsnm{Cao}, \binits{Y.}},
\bauthor{\bsnm{Cao}, \binits{Y.}},
\bauthor{\bsnm{Fu}, \binits{F.-W.}},
\bauthor{\bsnm{Gao}, \binits{J.}}:
\batitle{On a class of left metacyclic codes}.
\bjtitle{IEEE Transactions on Information Theory}
\bvolume{62}(\bissue{12}),
\bfpage{6786}--\blpage{6799}
(\byear{2016})
\doiurl{10.1109/tit.2016.2613115}
\end{barticle}
\endbibitem

\bibitem[\protect\citeauthoryear{Cao et~al.}{2022}]{Cao2022}
\begin{barticle}
\bauthor{\bsnm{Cao}, \binits{Y.}},
\bauthor{\bsnm{Cao}, \binits{Y.}},
\bauthor{\bsnm{Ma}, \binits{F.}}:
\batitle{Construction and enumeration of left dihedral codes satisfying certain duality properties}.
\bjtitle{Discrete Mathematics}
\bvolume{345}(\bissue{11}),
\bfpage{113059}
(\byear{2022})
\doiurl{10.1016/j.disc.2022.113059}
\end{barticle}
\endbibitem

\bibitem[\protect\citeauthoryear{Borello and Jamous}{2021}]{Borello2021}
\begin{bchapter}
\bauthor{\bsnm{Borello}, \binits{M.}},
\bauthor{\bsnm{Jamous}, \binits{A.}}:
\bctitle{Dihedral codes with prescribed minimum distance}.
In: \bbtitle{Arithmetic of Finite Fields: 8th International Workshop, WAIFI 2020, Rennes, France, July 6--8, 2020, Revised Selected and Invited Papers 8},
pp. \bfpage{147}--\blpage{159}
(\byear{2021}).
\bcomment{Springer}
\end{bchapter}
\endbibitem

\bibitem[\protect\citeauthoryear{Lally}{2003}]{Lally2003}
\begin{bchapter}
\bauthor{\bsnm{Lally}, \binits{K.}}:
\bctitle{Quasicyclic codes of index $l$ over $\mathbb{F}_q$ viewed as $\mathbb{F}_q$-submodules of $(\mathbb{F}_q[x]/\left< x^m-1 \right>)^l$}.
In: \bbtitle{Applied Algebra, Algebraic Algorithms and Error-Correcting Codes: 15th International Symposium, AAECC-15, Toulouse, France, May 12--16, 2003 Proceedings 15},
pp. \bfpage{244}--\blpage{253}
(\byear{2003}).
\bcomment{Springer}
\end{bchapter}
\endbibitem

\bibitem[\protect\citeauthoryear{Reiner}{2003}]{Reiner2003}
\begin{bbook}
\bauthor{\bsnm{Reiner}, \binits{I.}}:
\bbtitle{Maximal Orders}.
\bpublisher{Oxford University Press, USA}, \blocation{???}
(\byear{2003})
\end{bbook}
\endbibitem

\bibitem[\protect\citeauthoryear{Lang}{2002}]{Lang}
\begin{bbook}
\bauthor{\bsnm{Lang}, \binits{S.}}:
\bbtitle{Algebra},
\bedition{3}rd edn.
\bsertitle{Graduate Texts in Mathematics 211}.
\bpublisher{Springer}, \blocation{???}
(\byear{2002})
\end{bbook}
\endbibitem

\bibitem[\protect\citeauthoryear{Caruso and Drain}{2023}]{Caruso2023}
\begin{botherref}
\oauthor{\bsnm{Caruso}, \binits{X.}},
\oauthor{\bsnm{Drain}, \binits{F.}}:
{Selfdual skew cyclic codes}.
working paper or preprint
(2023).
\url{https://hal.science/hal-04127001}
\end{botherref}
\endbibitem

\bibitem[\protect\citeauthoryear{Morita}{1958}]{Morita1958}
\begin{barticle}
\bauthor{\bsnm{Morita}, \binits{K.}}:
\batitle{Duality for modules and its applications to the theory of rings with minimum condition}.
\bjtitle{Science Reports of the Tokyo Kyoiku Daigaku, Section A}
\bvolume{6}(\bissue{150}),
\bfpage{83}--\blpage{142}
(\byear{1958})
\end{barticle}
\endbibitem

\bibitem[\protect\citeauthoryear{Ferraz et~al.}{2021}]{Ferraz2021ideals}
\begin{barticle}
\bauthor{\bsnm{Ferraz}, \binits{R.A.}},
\bauthor{\bsnm{Milies}, \binits{C.P.}},
\bauthor{\bsnm{Taufer}, \binits{E.}}:
\batitle{Left ideals of matrix rings and error-correcting codes}.
\bjtitle{Applicable Algebra in Engineering, Communication and Computing}
\bvolume{32}(\bissue{3}),
\bfpage{311}--\blpage{320}
(\byear{2021})
\doiurl{10.1007/s00200-021-00498-4}
\end{barticle}
\endbibitem

\bibitem[\protect\citeauthoryear{Ding}{2021}]{Ding2021}
\begin{bchapter}
\bauthor{\bsnm{Ding}, \binits{C.}}:
\bctitle{Cyclic codes over finite fields}.
In: \bbtitle{Concise Encyclopedia of Coding Theory},
pp. \bfpage{45}--\blpage{60}.
\bpublisher{Chapman and Hall/CRC}, \blocation{???}
(\byear{2021})
\end{bchapter}
\endbibitem

\bibitem[\protect\citeauthoryear{Guneri et~al.}{2021}]{Guneri2021}
\begin{bchapter}
\bauthor{\bsnm{Guneri}, \binits{C.}},
\bauthor{\bsnm{Ling}, \binits{S.}},
\bauthor{\bsnm{Ozkaya}, \binits{B.}}:
\bctitle{Quasi-cyclic codes}.
In: \bbtitle{Concise Encyclopedia of Coding Theory},
pp. \bfpage{45}--\blpage{60}.
\bpublisher{Chapman and Hall/CRC}, \blocation{???}
(\byear{2021})
\end{bchapter}
\endbibitem

\bibitem[\protect\citeauthoryear{Gluesing-Luerssen}{2021}]{Gluesing2021}
\begin{bchapter}
\bauthor{\bsnm{Gluesing-Luerssen}, \binits{H.}}:
\bctitle{Introduction to skew-polynomial rings and skew-cyclic codes}.
In: \bbtitle{Concise Encyclopedia of Coding Theory},
pp. \bfpage{45}--\blpage{60}.
\bpublisher{Chapman and Hall/CRC}, \blocation{???}
(\byear{2021})
\end{bchapter}
\endbibitem

\bibitem[\protect\citeauthoryear{Blokh and Zyablov}{1974}]{Blokh1974}
\begin{barticle}
\bauthor{\bsnm{Blokh}, \binits{{\`E}.L.}},
\bauthor{\bsnm{Zyablov}, \binits{V.V.}}:
\batitle{Coding of generalized concatenated codes}.
\bjtitle{Problemy Peredachi Informatsii}
\bvolume{10}(\bissue{3}),
\bfpage{45}--\blpage{50}
(\byear{1974})
\end{barticle}
\endbibitem

\bibitem[\protect\citeauthoryear{Zyablov et~al.}{1999}]{Zyablov1999}
\begin{barticle}
\bauthor{\bsnm{Zyablov}, \binits{V.}},
\bauthor{\bsnm{Shavgulidze}, \binits{S.}},
\bauthor{\bsnm{Bossert}, \binits{M.}}:
\batitle{An introduction to generalized concatenated codes}.
\bjtitle{European Transactions on Telecommunications}
\bvolume{10}(\bissue{6}),
\bfpage{609}--\blpage{622}
(\byear{1999})
\doiurl{10.1002/ett.4460100606}
\end{barticle}
\endbibitem

\bibitem[\protect\citeauthoryear{Puchinger et~al.}{2017}]{Puchinger2017}
\begin{bbook}
\bauthor{\bsnm{Puchinger}, \binits{S.}},
\bauthor{\bsnm{M\"{u}elich}, \binits{S.}},
\bauthor{\bsnm{Ishak}, \binits{K.}},
\bauthor{\bsnm{Bossert}, \binits{M.}}:
\bbtitle{Code-Based Cryptosystems Using Generalized Concatenated Codes},
pp. \bfpage{397}--\blpage{423}.
\bpublisher{Springer}, \blocation{???}
(\byear{2017}).
\doiurl{10.1007/978-3-319-56932-1_26} .
\burl{http://dx.doi.org/10.1007/978-3-319-56932-1_26}
\end{bbook}
\endbibitem

\bibitem[\protect\citeauthoryear{Zimmermann}{1994}]{Zimmermann1994}
\begin{bbook}
\bauthor{\bsnm{Zimmermann}, \binits{K.-H.}}:
\bbtitle{Beitr{\"a}ge zur Algebraischen Codierungstheorie Mittels Modularer Darstellungstheorie}.
\bpublisher{Lehrstuhl II f{\"u}r Mathematik, Universit{\"a}t Bayreuth}, \blocation{???}
(\byear{1994})
\end{bbook}
\endbibitem

\bibitem[\protect\citeauthoryear{Lentmaier et~al.}{2010}]{Lentmaier2010}
\begin{bchapter}
\bauthor{\bsnm{Lentmaier}, \binits{M.}},
\bauthor{\bsnm{Liva}, \binits{G.}},
\bauthor{\bsnm{Paolini}, \binits{E.}},
\bauthor{\bsnm{Fettweis}, \binits{G.}}:
\bctitle{From product codes to structured generalized ldpc codes}.
In: \bbtitle{Proceedings of the 5th International ICST Conference on Communications and Networking in China}.
\bsertitle{CHINACOM}.
\bpublisher{IEEE}, \blocation{???}
(\byear{2010}).
\doiurl{10.4108/chinacom.2010.81} .
\burl{http://dx.doi.org/10.4108/chinacom.2010.81}
\end{bchapter}
\endbibitem

\end{thebibliography}

\end{document}